\documentclass[11pt]{article}
\usepackage[margin=1in]{geometry}
\usepackage[colorlinks=true, allcolors=blue]{hyperref}

\usepackage{fixmath}
\usepackage{bm}
\usepackage{amsbsy}
\usepackage{color}
\usepackage{verbatim}
\usepackage{multirow}
\usepackage{amssymb}
\usepackage{amsthm}
\usepackage{array}
\usepackage{mathtools}
\usepackage{amsmath}
\usepackage{graphicx}
\usepackage{tikz}
\usetikzlibrary{positioning}
\usepackage{xcolor}
\usepackage{thm-restate}

\newtheorem{theorem}{Theorem}
\newtheorem{lemma}[theorem]{Lemma}

\let\varepsilon=\varepsilon

\newcommand{\size}[1]{\ensuremath{|#1|}}
\newcommand{\ceil}[1]{\ensuremath{\lceil#1\rceil}}

\newcommand{\lrA}[1]{\ensuremath{\left(#1\right)}}

\def\B{\mathcal{B}}
\def\C{\mathcal{C}}
\def\M{\mathcal{M}}

\def\P{\mathcal{P}}
\def\T{\mathcal{T}}

\def\L{\mathcal{L}}
\def\R{\mathcal{R}}

\def\X{\mathcal{X}}
\def\Y{\mathcal{Y}}
\def\Z{\mathcal{Z}}

\newcommand{\PP}[1]{\ensuremath{\mbox{Pr}[#1]}}

\newcommand{\EE}[1]{\ensuremath{\mathbb{E}[#1]}}

\title{An Improved Approximation Algorithm for Metric Triangle Packing}

\author
{
Jingyang Zhao\footnote{University of Electronic Science and Technology of China. Email: \texttt{jingyangzhao1020@gmail.com}.}
\and
Mingyu Xiao\footnote{University of Electronic Science and Technology of China.
Email: \texttt{myxiao@gmail.com}.}
}

\date{}

\begin{document}

\maketitle

\begin{abstract}
Given an edge-weighted metric complete graph with $n$ vertices, the maximum weight metric triangle packing problem is to find a set of $n/3$ vertex-disjoint triangles with the total weight of all triangles in the packing maximized. Several simple methods can lead to a 2/3-approximation ratio. However, this barrier is not easy to break. Chen et al. proposed a randomized approximation algorithm with an expected ratio of $(0.66768-\varepsilon)$ for any constant $\varepsilon>0$. In this paper, we improve the approximation ratio to $(0.66835-\varepsilon)$. Furthermore, we can derandomize our algorithm.

\end{abstract}

\section{Introduction}
In a graph with $n$ vertices, a \emph{triangle packing} is a set of vertex-disjoint triangles (i.e., a simple cycle on three different vertices).
The triangle packing is called \emph{perfect} if its size is $n/3$ (i.e., it can cover all vertices).
Given an unweighted graph, the Maximum Triangle Packing (MTP) problem is to find a triangle packing of maximum cardinality.
In an edge-weighted graph, every edge has a non-negative weight.
There are two natural variants.
If the graph is an edge-weighted complete graph, a perfect triangle packing will exist. The Maximum Weight Triangle Packing (MWTP) problem is to find a perfect triangle packing such that the total weight of all triangles is maximized. Furthermore, if the graph is an edge-weighted metric complete graph (i.e., the weight of edges satisfies the symmetric and triangle inequality properties), the problem is called the Maximum Weight Metric Triangle Packing (MWMTP) problem.

In this paper, we mainly study approximation algorithms of MWMTP.

\subsection{Related work}
It is known~\cite{KirkpatrickH78} that even deciding whether an unweighted graph contains a perfect triangle packing is NP-hard. Hence, MTP, MWTP and MWMTP are all NP-hard. 
MTP also includes the well-known 3-dimensional matching problem as a special case~\cite{Garey79}.
Guruswam et al.~\cite{DBLP:conf/wg/GuruswamiRCCW98} showed that MTP remains NP-hard on chordal, planar, line, and total graphs. Moreover, MTP has been proved to be APX-hard even on graphs with maximum degree 4~\cite{DBLP:journals/ipl/Kann91,DBLP:journals/mst/RooijNB13}.
Chleb{\'{\i}}k and Chleb{\'{\i}}kov{\'{a}}~\cite{DBLP:conf/ciac/ChlebikC03} showed that MTP is NP-hard to approximate better than 0.9929. MTP is a special case of the unweighted 3-Set Packing problem, which admits an approximation ratio of $(2/3-\varepsilon)$~\cite{DBLP:journals/siamdm/HurkensS89,DBLP:conf/soda/Halldorsson95} and $(3/4-\varepsilon)$~\cite{DBLP:conf/focs/Cygan13,DBLP:conf/iscopt/FurerY14}.
For MTP on graphs with maximum degree 4, Manic and Wakabayashi~\cite{DBLP:journals/dm/ManicW08} proposed a 0.8333-approximation algorithm.

Similarly, MWTP can be seen as a special case of the weighted 3-Set Packing problem, which admits an approximation ratio of $(1/2-\varepsilon)$~\cite{arkin1998local,DBLP:journals/njc/Berman00}, $1/(2-1/63700992+\varepsilon)$~\cite{DBLP:conf/stacs/Neuwohner21}, and (1/1.786)~\cite{thiery2023improved}. For MWTP, there are some independent approximation algorithms. Hassin and Rubinstein~\cite{hassin2006approximation,DBLP:journals/dam/HassinR06a} proposed a randomized $(0.518-\varepsilon)$-approximation algorithm. Chen et al.~\cite{DBLP:journals/dam/ChenTW09,DBLP:journals/dam/ChenTW10} proposed an improved randomized $(0.523-\varepsilon)$-approximation algorithm. Using the method of pessimistic estimator, van Zuylen~\cite{DBLP:journals/dam/Zuylen13} proposed a deterministic algorithm with the same approximation ratio. The current best ratio is due to the $1/1.786$-approximation algorithm for the weighted 3-Set Packing problem~\cite{thiery2023improved}. For MWTP on $\{0,1\}$-weighted graphs (i.e., a complete graph with edge weights 0 and 1), Bar-Noy et al.~\cite{DBLP:journals/dam/Bar-NoyPRV18} proposed a $3/5$-approximation algorithm.

For MWMTP, Hassin et al.~\cite{hassin1997approximation1} gave the first deterministic $2/3$-approximation algorithm. Note that one can see that it is easy to design a $2/3$-approximation algorithm. Chen et al.~\cite{DBLP:journals/jco/ChenCLWZ21} proposed a nontrivial randomized approximation algorithm with an expected ratio of $(0.66768-\varepsilon)$. 

\subsection{Our results}
In this paper, we propose a deterministic $(0.66835-\varepsilon)$-approximation algorithm, which improves the deterministic $2/3$-approximation algorithm~\cite{hassin1997approximation1} and the randomized $(0.66768-\varepsilon)$-approximation algorithm~\cite{DBLP:journals/jco/ChenCLWZ21}. Our algorithm is based on the randomized algorithm in~\cite{DBLP:journals/jco/ChenCLWZ21}, but the framework of our analysis is completely different. The main differences are shown as follows.

Firstly, the previous algorithm considers so-called balanced/unbalanced triangles in an optimal solution of MWMTP, while in our algorithm we use novel definitions to make tighter analysis and do not need to separate them. Secondly, we also consider orientations of cycles in a cycle packing, which can simplify the structure significantly and enable us to design better algorithms. For example, the previous decomposition algorithm could only lead to a probability of $1/27$ (for one specific event) while our new decomposition algorithm enables us to obtain a probability of at least $97/1215$ (see details in Section~\ref{decomposition}). Lastly, our algorithm is deterministic, which is obtained by derandomizing our randomized algorithm: we first propose a new randomized algorithm such that if one specific matching is given our algorithm is deterministic, and then we further derandomize our algorithm by finding a desirable deterministic matching by the method of conditional exceptions (see details in Section~\ref{third}).

\section{Preliminaries}\label{prelim}
We use $G=(V, E)$ to denote a complete graph with $n$ vertices, where $n\bmod 3=0$. There is a non-negative weight function $w: E\to \mathbb{R}_{\geq0}$ on the edges in $E$.
The weight function $w$ is a semi-metric function, i.e., it satisfies the symmetric and triangle inequality properties.
For any weight function $w:X\to \mathbb{R}_{\geq0}$, we extend it to subsets of $X$, i.e., we define $w(Y) = \sum_{x\in Y} w(x)$ for $Y\subseteq X$.

A triangle $t=xyz$ is a simple cycle on three different vertices $\{x,y,z\}$. It contains exactly three edges $\{xy,xz,yz\}$. We may also use $\{a_t,b_t,c_t\}$ to denote them such that $w(a_t)\leq w(b_t)\leq w(c_t)$.

Two subgraphs or sets of edges are \emph{vertex-disjoint} if they do not share a common vertex. As mentioned, a \emph{perfect triangle packing} in graph $G$ is a set of vertex-disjoint $n/3$ triangles, and all vertices are covered.
It can be seen as the edges of $n/3$ vertex-disjoint triangles.
In the following, we will always consider a triangle packing as a perfect triangle packing. We will use $\B^*$ to denote the maximum weight triangle packing.

A \emph{matching} is a set of vertex-disjoint edges. In this paper, we often consider a matching of size $n/3$ and use $\M^*$ to denote the maximum weight matching, which can be found in $O(n^3)$ time~\cite{gabow1974implementation,lawler1976combinatorial}. A \emph{cycle packing} is a set of vertex-disjoint simple cycles, the length of each cycle is at least 3, and all vertices of the graph are covered. We use $\C^*$ to denote the maximum weight cycle packing, which can be found in $O(n^3)$ time~\cite{hartvigsen1984extensions}. We can get $w(\C^*)\geq w(\B^*)$ since a triangle packing is also a cycle packing. Given a set of edges $\X$ and a vertex $x$, we may simply use $x\in\X$ to denote that there is an edge in $\X$ that contains the vertex $x$. 

Let $\varepsilon$ be a constant such that $0<\varepsilon\leq 2/5$. A cycle $C$ is \emph{short} if $\size{C}\leq \frac{2}{\varepsilon}$; otherwise, the cycle is \emph{long}. For each long cycle $C\in\C^*$, it is easy to see that we can delete at most $\ceil{\frac{\varepsilon}{2}\size{C}}<\frac{\varepsilon}{2}\size{C}+1<\varepsilon\size{C}$ edges to get a set of paths such that the length of each path, $\size{P}$, satisfies that $3\leq\size{P}\leq\frac{2}{\varepsilon}$.
Moreover, the total weight of these paths is at least $(1-\varepsilon)w(C)$.
By connecting the endpoints of each path, we can get a set of short cycles. Hence, we can get a short cycle packing (i.e., a cycle packing containing only short cycles) in polynomial time, denoted by $\C$, such that $w(\C)\geq(1-\varepsilon)w(\C^*)$. Note that the constant $\varepsilon$ needs to be at most $2/5$. Otherwise, even the cycle of length 5 is a long cycle, and then the short cycle packing may contain a cycle of length less than 3.

We first define several kinds of triangles in $\B^*$.
Fix the short cycle packing $\C$. An edge is an \emph{internal edge} if both of its endpoints fall on the same cycle in $\C$; otherwise, it is an \emph{external edge}. Consider a triangle $t$ in the optimal triangle packing $\B^*$.
There are three cases:
\begin{itemize}
    \item $t$ is an \emph{internal} triangle if it contains three internal edges;
    \item $t$ is a \emph{partial-external} triangle if it contains one internal edge and two external edges;
    \item $t$ is an \emph{external} triangle if it contains three external edges.
\end{itemize}
See Figure~\ref{fig01} for an illustration.

\begin{figure}[ht]
\centering
\begin{tikzpicture}
\filldraw [black]
(-1,0) circle [radius=2pt]
(0,0) circle [radius=2pt]
(-1,1) circle [radius=2pt]
(0,1) circle [radius=2pt]
(-1,2) circle [radius=2pt]
(0,2) circle [radius=2pt]
(-1,3) circle [radius=2pt]
(0,3) circle [radius=2pt]

(3,0) circle [radius=2pt]
(3,1) circle [radius=2pt]
(3,2) circle [radius=2pt]
(4,0) circle [radius=2pt]
(4,2) circle [radius=2pt]
(5,0) circle [radius=2pt]
(5,2) circle [radius=2pt]
(6,0) circle [radius=2pt]
(6,1) circle [radius=2pt]
(6,2) circle [radius=2pt]

(2,3) circle [radius=2pt]
(2,4) circle [radius=2pt]
(2,5) circle [radius=2pt]
(3,3) circle [radius=2pt]
(3,5) circle [radius=2pt]
(4,3) circle [radius=2pt]
(4,5) circle [radius=2pt]
(5,3) circle [radius=2pt]
(5,4) circle [radius=2pt]
(5,5) circle [radius=2pt];

\draw[very thick,dotted] (0,0) to (0,3);
\draw[very thick,dotted] (0,0) to (-1,0);
\draw[very thick,dotted] (-1,0) to (-1,3);
\draw[very thick,dotted] (-1,3) to (0,3);

\draw[very thick,dotted] (2,3) to (2,5);
\draw[very thick,dotted] (2,5) to (5,5);
\draw[very thick,dotted] (5,5) to (5,3);
\draw[very thick,dotted] (5,3) to (2,3);

\draw[very thick,dotted] (3,0) to (3,2);
\draw[very thick,dotted] (3,2) to (6,2);
\draw[very thick,dotted] (6,2) to (6,0);
\draw[very thick,dotted] (6,0) to (3,0);

\draw[very thick] (0,1) to (2,3);
\draw[very thick] (3,1) to (2,3);
\draw[very thick] (0,1) to (3,1);
\node at (5/3,5/3) {\small $t_1$};

\draw[very thick] (3,0) to (6,2);
\draw[very thick] (6,2) to (4,3);
\draw[very thick] (4,3) to (3,0);
\node at (13/3,5/3) {\small $t_2$};

\draw[very thick] (3,3) to (5,4);
\draw[very thick] (5,4) to (3,5);
\draw[very thick] (3,5) to (3,3);
\node at (11/3,12/3) {\small $t_3$};
\end{tikzpicture}
\caption{An illustration of the three kinds of triangles, where there are three cycles (the dotted edges) in $\C$ and three triangles (the solid edges) in $\B^*$: $t_1$ is an external triangle, $t_2$ is a partial-external triangle, and $t_3$ is an internal triangle}
\label{fig01}
\end{figure}
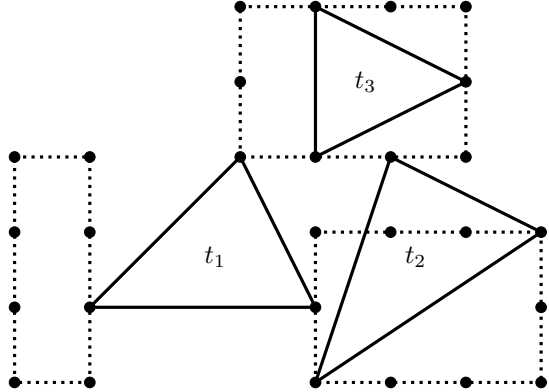

Given the short cycle packing $\C$, for the sake of analysis, we orient each cycle with an arbitrary direction (note that the cycles in~\cite{DBLP:journals/jco/ChenCLWZ21} are not oriented). Then, $\C$ becomes a set of directed cycles. For a vertex $x\in t\in \B^*$, it is an \emph{external vertex} if it is incident to two external edges of $t$. Hence, a partial-external triangle contains only one external vertex, and an external triangle contains three external vertices. For an external vertex $x\in t\in\B^*$, assume $x\in C\in\C$, the neighbor of $x$ on the directed cycle $C$ is denoted by $x'$, and then the directed edge $xx'$ is called an \emph{out-edge} of triangle $t$. The set of out-edges of triangle $t$ is denoted by $E_t$. Analogously, for each partial-external triangle $\size{E_t}=1$, for each external triangle $\size{E_t}=3$, and for each internal triangle $\size{E_t}=0$.
See Figure~\ref{fig02} for an illustration.

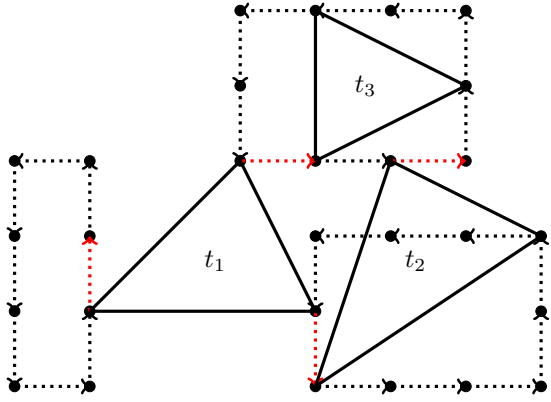
\begin{figure}[ht]
\centering
\begin{tikzpicture}
\filldraw [black]
(-1,0) circle [radius=2pt]
(0,0) circle [radius=2pt]
(-1,1) circle [radius=2pt]
(0,1) circle [radius=2pt]
(-1,2) circle [radius=2pt]
(0,2) circle [radius=2pt]
(-1,3) circle [radius=2pt]
(0,3) circle [radius=2pt]

(3,0) circle [radius=2pt]
(3,1) circle [radius=2pt]
(3,2) circle [radius=2pt]
(4,0) circle [radius=2pt]
(4,2) circle [radius=2pt]
(5,0) circle [radius=2pt]
(5,2) circle [radius=2pt]
(6,0) circle [radius=2pt]
(6,1) circle [radius=2pt]
(6,2) circle [radius=2pt]

(2,3) circle [radius=2pt]
(2,4) circle [radius=2pt]
(2,5) circle [radius=2pt]
(3,3) circle [radius=2pt]
(3,5) circle [radius=2pt]
(4,3) circle [radius=2pt]
(4,5) circle [radius=2pt]
(5,3) circle [radius=2pt]
(5,4) circle [radius=2pt]
(5,5) circle [radius=2pt];

\draw[very thick,dotted,->] (0,0) to (0,1);

\draw[very thick,dotted,->,red] (0,1) to (0,2);

\draw[very thick,dotted,->] (0,2) to (0,3);

\draw[very thick,dotted,<-] (0,0) to (-1,0);

\draw[very thick,dotted,<-] (-1,0) to (-1,1);

\draw[very thick,dotted,<-] (-1,1) to (-1,2);

\draw[very thick,dotted,<-] (-1,2) to (-1,3);

\draw[very thick,dotted,<-] (-1,3) to (0,3);

\draw[very thick,dotted,<-] (2,3) to (2,4);
\draw[very thick,dotted,<-] (2,4) to (2,5);

\draw[very thick,dotted,<-] (2,5) to (3,5);
\draw[very thick,dotted,<-] (3,5) to (4,5);
\draw[very thick,dotted,<-] (4,5) to (5,5);

\draw[very thick,dotted,<-] (5,5) to (5,4);
\draw[very thick,dotted,<-] (5,4) to (5,3);

\draw[very thick,dotted,<-,red] (5,3) to (4,3);
\draw[very thick,dotted,<-] (4,3) to (3,3);
\draw[very thick,dotted,<-,red] (3,3) to (2,3);

\draw[very thick,dotted,<-,red] (3,0) to (3,1);
\draw[very thick,dotted,<-] (3,1) to (3,2);

\draw[very thick,dotted,<-] (3,2) to (4,2);
\draw[very thick,dotted,<-] (4,2) to (5,2);
\draw[very thick,dotted,<-] (5,2) to (6,2);

\draw[very thick,dotted,<-] (6,2) to (6,1);
\draw[very thick,dotted,<-] (6,1) to (6,0);

\draw[very thick,dotted,<-] (6,0) to (5,0);
\draw[very thick,dotted,<-] (5,0) to (4,0);
\draw[very thick,dotted,<-] (4,0) to (3,0);

\draw[very thick] (0,1) to (2,3);
\draw[very thick] (3,1) to (2,3);
\draw[very thick] (0,1) to (3,1);
\node at (5/3,5/3) {\small $t_1$};

\draw[very thick] (3,0) to (6,2);
\draw[very thick] (6,2) to (4,3);
\draw[very thick] (4,3) to (3,0);
\node at (13/3,5/3) {\small $t_2$};

\draw[very thick] (3,3) to (5,4);
\draw[very thick] (5,4) to (3,5);
\draw[very thick] (3,5) to (3,3);
\node at (11/3,12/3) {\small $t_3$};
\end{tikzpicture}
\caption{An illustration of the oriented cycles (the dotted directed edges) and out-edges (the red dotted directed edges), where there are three out-edges from $t_1$ (an external triangle), one out-edge from $t_2$ (a partial-external triangle), and no out-edges from $t_3$ (an internal triangle)
}
\label{fig02}
\end{figure}

Next, we further define type-1 and type-2 triangles. Fix a constant $\tau$ with $0\leq \tau\leq 1/3$. For each triangle $t\in\B^*$, it is a \emph{type-1} triangle if it holds that $w(e)> \frac{1}{2}(1-\tau)w(t)$ for every edge $e\in E_t$; otherwise, it is a \emph{type-2} triangle (note that our definition is different from the definition in \cite{DBLP:journals/jco/ChenCLWZ21}).
For this definition, we only consider partial-external triangles and external triangles since internal triangles have no out-edges. Intuitively, the out-edges of type-1 triangles are all heavy edges.

In the following, we split $\B^*$ into five disjoint sets $\B^*_1$, $\B^*_2$, $\B^*_3$, $\B^*_4$, and $\B^*_5$ such that
\begin{enumerate}
    \item[$\B^*_1$:] the set of all internal triangles in $\B^*$;
    \item[$\B^*_2$:] the set of all type-1 partial-external triangles in $\B^*$;
    \item[$\B^*_3$:] the set of all type-2 partial-external triangles in $\B^*$;
    \item[$\B^*_4$:] the set of all type-1 external triangles in $\B^*$;
    \item[$\B^*_5$:] the set of all type-2 external triangles in $\B^*$.
\end{enumerate}
Note that $\B^*_1$ is the set of all internal triangles in $\B^*$, $\B^*_2\cup \B^*_3$ is the set of all partial-external triangles in $\B^*$, and $\B^*_4\cup \B^*_5$ is the set of all external triangles in $\B^*$. Hence, we have $\B^*=\B^*_1\cup\B^*_2\cup\B^*_3\cup\B^*_4\cup \B^*_5$. 

We also propose two new important definitions, which will be used frequently:
\[
u_i=\sum_{t\in \B^*_i}w(a_t)/\sum_{t\in \B^*_i}w(c_t)\quad\mbox{and}\quad v_i=\sum_{t\in \B^*_i}w(b_t)/\sum_{t\in \B^*_i}w(c_t).
\]
So, $\sum_{t\in\B^*_i}w(a_t)=\frac{u_i}{u_i+v_i+1}w(\B^*_i)$, $\sum_{t\in\B^*_i}w(b_t)=\frac{v_i}{u_i+v_i+1}w(\B^*_i)$, and $\sum_{t\in\B^*_i}w(c_t)=\frac{1}{u_i+v_i+1}w(\B^*_i)$.

\subsection{Paper Organization}
In the remaining parts of the paper, we will describe the approximation algorithm. It will compute three triangle packings: $\T_1$, $T_2$, and $\T_3$. The computation of $\T_1$, which can be seen in Section~\ref{first}, is based on the short cycle packing $\C$ via a dynamic method. In Section~\ref{second}, we will use the maximum weight matching $\M^*$ to compute $\T_2$. In Section~\ref{third}, we will first compute a randomized matching on the short cycle packing $\C$, and then use it to compute $\T_3$. The algorithm of $\T_3$ is randomized but we show that it can be derandomized efficiently by the method of conditional expectations. The approximation algorithm will return the best one. Hence, the approximation ratio is
\[
\frac{\max\{w(\T_1),\ w(T_2),\ w(\T_3)\}}{w(\B^*)}.
\]
The trade-off among these three triangle packings will be shown in Section~\ref{analysis}.

\section{The First Triangle Packing}\label{first}
\subsection{The algorithm}
A \emph{partial triangle packing}, denoted by $\P$, is a set of triangle-components and edge-components such that the total number of components is at most $n/3$ and the edges of each component are all internal edges.
The augmented weight of $\P$ is defined as $\widetilde{w}(\P)=\sum_{t\in\P}w(t)+\sum_{e\in \P}2w(e)$.
Suppose there are $p_1$ triangle-components and $p_2$ edge-components in $\P$, we have $n-3p_1-2p_2\geq p_2$ unused vertices (i.e., vertices not in $\P$).
Hence there are at least as many unused vertices as edge-components.
Given a partial triangle packing $\P$, one can construct a triangle packing $\T$ as follows:

\medskip
\noindent\textbf{Step~1.} Pick all triangle-components into the packing;

\noindent\textbf{Step~2.} Arbitrarily assign an unused vertex for each edge-component and complete them into a triangle;

\noindent\textbf{Step~3.} Arbitrarily construct a set of triangles if there are still unused vertices.
\medskip

Note that the weight of triangles in {Step~1} is $\sum_{t\in\P}w(t)$ and the weight of triangles in {Step~2} is at least $\sum_{e\in\P}2w(e)$ by the triangle inequality. Hence, we have that $w(\T)\geq \widetilde{w}(\P)$.

Since the length of each cycle in $\C$ is bounded by a constant, the maximum augmented weight partial triangle packing, denoted by $\P^*$, can be found in polynomial time via a dynamic  programming method~\cite{DBLP:journals/jco/ChenCLWZ21}. Hence, we can get the following lemma.

\begin{lemma}[\cite{DBLP:journals/jco/ChenCLWZ21}]\label{dp}
There is a polynomial-time algorithm that can compute a triangle packing $\T_1$ such that $w(\T_1)\geq \widetilde{w}(\P^*)$.
\end{lemma}


\subsection{The analysis}
Let $p_1$ and $p_2$ be the number of internal triangles and partial-external triangles, respectively.
After deleting all external edges of all triangles in $\B^*$, we can get a set of components, denoted by $\P$, where there are exactly $p_1$ triangle-components and $p_2$ edge-components. We can get $p_1+p_2\leq n/3$ since a triangle-component corresponds to an internal triangle and an edge-component corresponds to a partial-external triangle. Hence, $\P$ is a partial triangle packing. We can get that 
\[
w(\T_1)\geq \widetilde{w}(\P^*)\geq\widetilde{w}(\P)=\sum_{t\in\P}w(t)+\sum_{e\in \P}2w(e).
\]
It is easy to see that $\sum_{t\in\P}w(t)=w(\B^*_1)$ since the triangle-components in $\P$ contains all internal triangles and $\B_1$ is the set of all internal triangles. However, for edge-components, it contains only one (internal) edge of each partial-external triangle. In the worst case, the contained edge in each partial-external triangle $t$ is the least weighted edge $a_t$. So, we have $\sum_{e\in \P}2w(e)\geq\sum_{t\in\B^*_2\cup\B^*_3}2w(a_t)$. 
Then, we can get that 
\[
w(\T_1)\geq w(\B^*_1)+\sum_{t\in\B^*_2\cup\B^*_3}2w(a_t).
\]
Recall that $\sum_{t\in\B^*_i}w(a_t)=\frac{u_i}{u_i+v_i+1}w(\B^*_i)$. By Lemma~\ref{dp}, we can get the following lemma.

\begin{lemma}\label{t1}
There is a polynomial-time algorithm that can compute a triangle packing $\T_1$ such that
\[
w(\T_1)\geq w(\B^*_1)+\sum_{i=2}^{3}\frac{2u_{i}}{u_{i}+v_{i}+1}w(\B^*_{i}).
\]
\end{lemma}

\section{The Second Triangle Packing}\label{second}
\subsection{The algorithm}
The second triangle packing $\T_2$ is generated using the maximum weight matching $\M^*$. The algorithm is simple and contains two following steps.

\medskip
\noindent\textbf{Step~1.} Find the maximum weight matching $\M^*$ (of size $n/3$) in $O(n^3)$ time~\cite{gabow1974implementation,lawler1976combinatorial};

\noindent\textbf{Step~2.} Arbitrarily assign an unused vertex for each edge of $\M^*$ and complete them into a triangle.
\medskip

The matching $\M^*$ uses $2n/3$ vertices, and hence after {Step~1}, there are still $n/3$ unused vertices. Then, we can get a triangle packing (i.e., $\T_2$) after {Step~2}. Moreover, it is easy to see that the running time is $O(n^3)$, dominated by computing $\M^*$.

Note that we can also construct a bipartite graph $B=(\X, V\setminus V(\X))$ such that there are $n/3$ super-vertices on the left corresponding to the $n/3$ edges in $\X$ and $n/3$ super-vertices on the right corresponding to the unused $n/3$ vertices in $V\setminus V(\X)$, where $V(\X)$ denotes the set of vertices used by $\X$. For each left super-vertex $xy\in \X$ and each right super-vertex $z\in V\setminus V(\X)$, the weight of the edge between them is defined as $w(xz)+w(yz)$. We could obtain a better triangle packing that contains the matching $\M^*$ by finding a maximum weight matching in graph $B$, as did in~\cite{DBLP:journals/jco/ChenCLWZ21}. But, it cannot improve the approximation ratio in our analysis. So, we only use the simpler algorithm.

\subsection{The analysis}
Consider a triangle $t=xyz\in \T_2$, and assume w.l.o.g. that $xy\in\M^*$. By the triangle inequality, we can get that $w(t)=w(xy)+w(xz)+w(yz)\geq 2w(xy)$. Hence, it is easy to get $w(\T_2)\geq 2w(\M^*)$.

Take the most weighted edge $c_t$ for each triangle $t\in \B^*$. Then, the set of edges forms a matching of size $n/3$. Since $\M^*$ is the maximum weight matching of size $n/3$, we have $w(\M^*)\geq\sum_{t\in\B^*}w(c_t)$. Recall that $\sum_{t\in\B^*_i}w(c_t)=\frac{1}{u_i+v_i+1}w(\B^*_i)$ and $w(\T_2)\geq 2w(\M^*)$. We can get the following lemma.
\begin{lemma}\label{t2}
There is a polynomial-time algorithm that can compute a triangle packing $\T_2$ such that
\[
w(\T_2)\geq\sum_{t\in\B^*}2w(c_t)=\sum_{i=1}^{5}\frac{2}{u_i+v_i+1}w(\B^*_i).
\]
\end{lemma}

It is worth noting that that the algorithm of $\T_2$ is a simple $2/3$-approximation algorithm. Since  $u_i\leq v_i\leq 1$, we can get $w(\T_2)\geq \frac{2}{3}w(\B^*)$ by Lemma~\ref{t2}.
Next, we will construct a randomized matching $\Z$ of size $n/3$. A good advantage of $\Z$ is that if type-1 triangles in $\B^*$ has a nonzero weight, the expected weight of $\Z$ will have a strictly larger weight than $\frac{1}{3}w(\B^*)$. In this case, $\T_2$ will have a strictly larger weight than $\frac{2}{3}w(\B^*)$ since $\M^*$ is the maximum weight matching and we have $w(\T_2)\geq 2w(\M^*)\geq 2\cdot\EE{w(\Z)}$.

\subsection{The randomized matching algorithm}~\label{decomposition}
Our randomized matching algorithm is a refined version of the algorithm presented in~\cite{DBLP:journals/jco/ChenCLWZ21}, which mainly contains four following steps.

\medskip
\noindent\textbf{Step~1.} For each triangle $t\in\B^*$, select an edge $e_t$ uniformly at random. The set of selected edges is denoted by $\X$.

\noindent\textbf{Step~2.} For each type-1 triangle $t\in\B^*$, if there is an out-edge that does not share a common vertex with any edge of $\X$, select it. The set of selected edges is denoted by $\Y$.

\noindent\textbf{Step~3.} Initialize $\Z=\emptyset$. Consider the out-edges in $\Y$, which contains a set of edge-components, path-components, and cycle-components. We consider three following cases.
\begin{itemize}
    \item For each edge-component, select the edge into $\Z$.
    \item For each path-component or even cycle-component (i.e., the number of vertices in it is even), partition it into two matchings, select one matching uniformly at random, and then select the edges in the chosen matching into $\Z$.
    \item For each odd cycle-component (i.e., the number of vertices in it is odd), select one edge uniformly at random, delete it, partition it into two matchings, select one matching uniformly at random, and then select the edges in the chosen matching into $\Z$.
\end{itemize} 
Note that $\Z$ is a matching containing only out-edges. Moreover, the out-edges are vertex-disjoint with the edges in $\X$.

\noindent\textbf{Step~4.} Consider each triangle $t\in\B^*$. If $\Z$ contains no out-edges of $t$, select the edge $e_t$ in $\X$ into $\Z$.
\medskip

Roughly speaking, the main idea of the randomized matching algorithm in~\cite{DBLP:journals/jco/ChenCLWZ21} is to decompose the out-edges in $\Y$ into three matchings and put one into $\Z$, while we can put more out-edges into $\Z$ in {Step~3} in our algorithm. Although our algorithm becomes more complicated, we can quickly analyze the expected weight of the randomized matching $\Z$ because the cycles in $\C$ are all oriented.

\begin{lemma}
The edge set $\Z$ is a matching of size $n/3$.
\end{lemma}
\begin{proof}
First, it is easy to see that the edge set $\X$ in {Step~1} is a matching of size $n/3$. Second, the edge set $\Y$ in {Step~2} is a set of out-edges, which can be seen as a subgraph of the cycle packing $\C$. Hence, it contains a set of edge-components, path-components, and cycle-components, and {Step~3} is feasible. Then, since the out-edges in $\Y$ do not share a common vertex with any edge in $\X$, and $\Z$ after {Step~3} is a subset of vertex-disjoint edges in $\Y$, hence $\Z$ is also a matching. Moreover, it is easy to see that for each triangle $t\in\B^*$ the edge set $\Y$ contains at most one out-edge of $t$. Hence, according to {Step~4}, for each triangle $t\in\B^*$, $\Z$ either contains exactly one out-edge of $t$ or it will contain the edge $e_t$. Hence, there are exactly $n/3$ edges in $\Z$.
\end{proof}

Next, we analyze the expected weight of the randomized matching $\Z$. 

\begin{lemma}\label{e}
It holds that
\begin{align*}
\EE{w(\Z)}>&\ \frac{1}{3}w(\B^*)+\sum_{t\in\B^*_{2},e\in E_t}\lrA{\frac{1-3\tau}{6}w(t)\cdot\PP{e\in\Z}} +\sum_{t\in\B^*_{4}}\lrA{\frac{1-3\tau}{6}w(t)\cdot\sum_{e\in E_t}\PP{e\in\Z}}.
\end{align*}
\end{lemma}
\begin{proof}
Recall that for each triangle $t\in\B^*$, $\Z$ either contains exactly one edge of $t$ or exactly one out-edge of $t$. Hence, we will analyze the expected contributed weight of each triangle $t\in\B^*$. We will consider the following three cases.

\textbf{Case~1: Non type-1 triangles.}
For each triangle $t\in\B^*\setminus(\B^*_{2}\cup\B^*_{4})$ (i.e., the set of non type-1 triangles), it does not have out-edges, and hence it is clear that $e_t\in\Z$. Recall that $e_t$ was selected uniformly at random. Hence, the expected contributed weight is $\EE{w(e_t)}=\frac{1}{3}w(t)$. The total expected contributed weight of all non type-1 triangles is 
\[
\frac{1}{3}w(\B^*\setminus(\B^*_{2}\cup\B^*_{4}))=\frac{1}{3}w(\B^*)-\frac{1}{3}w(\B^*_2)-\frac{1}{3}w(\B^*_4).
\]

\textbf{Case~2: Type-1 partial-external triangles.}
For each type-1 partial-external triangle $t\in\B^*_{2}$, there is only one out-edge in $E_t$, denoted by $e$, and hence the expected contributed weight is $w(e)\cdot\PP{e\in\Z}+\EE{w(e_t)}\cdot(1-\PP{e\in\Z})$. Recall that for each type-1 triangle $t'$ it holds that $w(e)>\frac{1}{2}(1-\tau)w(t')$ for every edge $e\in E_{t'}$. Moreover, since $\EE{w(e_t)}=\frac{1}{3}w(t)$, we know that the expected contributed weight satisfies 
\begin{align*}
w(e)\cdot\PP{e\in\Z}+\EE{w(e_t)}\cdot(1-\PP{e\in\Z})
&>\frac{1}{2}(1-\tau)w(t)\cdot\PP{e\in\Z}+\frac{1}{3}w(t)\cdot(1-\PP{e\in\Z})\\
&=\frac{1}{3}w(t)+\frac{1-3\tau}{6}w(t)\cdot\PP{e\in\Z}.
\end{align*}
The total expected contributed weight of all type-1 partial-external triangles is at least 
\[
\frac{1}{3}w(\B^*_{2})+\sum_{t\in\B^*_{2},e\in E_t}\lrA{\frac{1-3\tau}{6}w(t)\cdot\PP{e\in\Z}}.
\]

\textbf{Case~3: Type-1 external triangles.}
For each type-1 external triangle $t\in\B^*_{4}$, there are three out-edge in $E_t$, and hence the expected contributed weight is $\sum_{e\in E_t}(w(e)\cdot\PP{e\in\Z})+\EE{w(e_t)}\cdot(1-\sum_{e\in E_t}\PP{e\in\Z})$.
By a similar argument with {Case~2}, the expected contributed weight satisfies 
\[
\sum_{e\in E_t}(w(e)\cdot\PP{e\in\Z})+\EE{w(e_t)}\cdot\lrA{1-\sum_{e\in E_t}\PP{e\in\Z}}> \frac{1}{3}w(t)+\frac{1-3\tau}{6}w(t)\cdot\sum_{e\in E_t}\PP{e\in\Z}.
\]
The total expected contributed weight of all type-1 external triangles is at least 
\[
\frac{1}{3}w(\B^*_{4})+\sum_{t\in\B^*_{4}}\lrA{\frac{1-3\tau}{6}w(t)\cdot\sum_{e\in E_t}\PP{e\in\Z}}.
\]

Therefore, we have that
\begin{align*}
\EE{w(\Z)}>&\ \frac{1}{3}w(\B^*)+\sum_{t\in\B^*_{2},e\in E_t}\lrA{\frac{1-3\tau}{6}w(t)\cdot\PP{e\in\Z}} +\sum_{t\in\B^*_{4}}\lrA{\frac{1-3\tau}{6}w(t)\cdot\sum_{e\in E_t}\PP{e\in\Z}}.
\end{align*}
\end{proof}

Next, we give a lower bound of $\PP{e\in\Z}$ for each type-1 triangle $t$ with $e\in E_t$.

\begin{lemma}\label{p}
For each type-1 triangle $t$ with any $e\in E_t$, it holds that $\PP{e\in \Z}\geq \frac{97}{1215}$.
\end{lemma}
\begin{proof}
We take an arbitrary type-1 triangle $t$ and then consider an arbitrary out-edge $e\in E_t$.
We assume that the out-edge $e$ is on the cycle $C\in\C$ and the corresponding segment of $C$ is $(x,y,z,k)$, where $e=yz$. Note that $x=k$ if $\size{C}=3$. For the sake of presentation, the triangle in $\B^*$ containing a vertex $v$ is denoted by $t_v$. Hence, we have $t=t_y$. See Figure~\ref{fig03} for an illustration.

\begin{figure}[ht]
\centering
\begin{tikzpicture}
\tikzstyle{vertex}=[black,circle,fill,minimum size=5,inner sep=0pt]
\begin{scope}[every node/.style={vertex}]
 \node (x) at (-1,0) {};
 \node (y) at (0,0) {};
 \node (z) at (1,0) {};
 \node (k) at (2,0) {};
\end{scope}

\node[below] at (-1,0) {$x$};
\node[below] at (0,0) {$y$};
\node[below] at (1,0) {$z$};
\node[below] at (2,0) {$k$};
\draw[very thick,-] (y.north) to (0.25,1.25);
\draw[very thick,-] (y.north) to (-0.25,1.25);

\node at (0,1) {\small $t_y$};

\draw[very thick,dotted,->] (x.east) to (y.west);
\draw[very thick,dotted,->,red] (y.east) to (z.west);
\draw[very thick,dotted,->] (z.east) to (k.west);
\end{tikzpicture}
\caption{An illustration of the type-1 triangle $t_y$ and one of its out-edge $e=yz$, where the dotted directed edges are the edges of cycle $C$}
\label{fig03}
\end{figure}

Since $e$ is an out-edge, by the definition, $y$ is an external vertex, and hence $t_y\neq t_z$.
We have $\PP{e\in \Y}= \PP{y\notin \X}\cdot \PP{z\notin \X}=\frac{1}{3}\cdot\frac{1}{3}=\frac{1}{9}$.
Next, we will show that $\PP{e\in \Z\mid e\in \Y}\geq \frac{97}{135}$.

\textbf{Case~1: both $x$ and $z$ are not external vertices.} In this case, we know that $xy\notin\Y$ and $zk\notin\Y$.
If $e\in \Y$, the component containing $e$ is an edge. Hence, $\PP{e\in \Z\mid e\in \Y}=1$.

\textbf{Case~2: one of $x$ and $z$ is an external vertex.}
We assume w.l.o.g. that $x$ is the external vertex and $z$ is not. We have $t_x\neq t_y$ and $zk\notin\Y$. So, $\PP{xy\in \Y\mid e\in \Y}=\frac{1}{3}$.
If $e\in \Y$ and $xy\in \Y$, the component containing $e$ is a path, and then $\PP{e\in \Z\mid e\in \Y \wedge xy\in \Y}=\frac{1}{2}$. If $e\in \Y$ and $xy\notin \Y$, the component containing $e$ is an edge, and then $\PP{e\in \Z\mid e\in \Y \wedge xy\notin \Y}=1$.
Hence, 
\begin{align*}
\PP{e\in \Z\mid e\in \Y}=&\ \PP{e\in \Z\mid xy\in \Y \wedge e\in \Y}\cdot\PP{xy\in \Y\mid e\in \Y}\\
&\  +\PP{e\in \Z\mid xy\notin \Y \wedge e\in \Y}\cdot\PP{xy\notin \Y\mid e\in \Y}\\
=&\ \frac{1}{2}\cdot\frac{1}{3}+1\cdot\frac{2}{3}\\
=&\ \frac{5}{6}.
\end{align*}

\textbf{Case~3: both $x$ and $z$ are external vertices, and $\size{C}=3$.}
We know that $t_x$, $t_y$, and $t_z$ are three different triangles, and then we have $\PP{xy\in \Y\mid e\in \Y}=\frac{1}{3}$. If $e\in \Y$ and $xy\in \Y$, we have $zx\in \Y$, and hence the component containing $e$ is an odd cycle. So, we have $\PP{e\in \Z\mid xy\in \Y\wedge e\in \Y }=\frac{2}{3}\cdot\frac{1}{2}=\frac{1}{3}$.
If $e\in \Y$ and $xy\notin \Y$, we have $zx\notin \Y$, and then the component containing $e$ is an edge. So, we have $\PP{e\in \Z\mid xy\notin \Y\wedge e\in \Y }=1$.
Hence, 
\begin{align*}
\PP{e\in \Z\mid e\in \Y}=&\ \PP{e\in \Z\mid xy\in \Y \wedge  e\in \Y}\cdot\PP{xy\in \Y\mid e\in \Y}\\
&\  +\PP{e\in \Z\mid xy\notin \Y \wedge  e\in \Y}\cdot\PP{xy\notin \Y\mid e\in \Y}\\
=&\ \frac{1}{3}\cdot\frac{1}{3}+1\cdot\frac{2}{3}\\
=&\ \frac{7}{9}.
\end{align*}

\textbf{Case~4: both $x$ and $z$ are external vertices, and $\size{C}\bmod2=0$.} Note that $\size{C}\geq 4$.
We know that $t_x$, $t_y$, $t_z$, and $t_k$ are four different triangles, and then we have $\PP{xy,zk\notin \Y\mid e\in \Y}=\frac{2}{3}\cdot\frac{2}{3}=\frac{4}{9}$.
If $e\in \Y$ and $xy,zk\notin \Y$, the component containing $e$ is an edge, and then we have $\PP{e\in \Z\mid xy,zk\notin \Y\wedge e\in \Y }=1$.
If $e\in \Y$ and $xy\in \Y \vee zk\in \Y$, the component containing $e$ is a path or an even cycle, and then we have $\PP{e\in \Z\mid (xy\in \Y \vee zk\in \Y)\wedge e\in \Y }=\frac{1}{2}$.
Note that $\PP{xy\in \Y \vee zk\in \Y\mid e\in \Y}=1-\PP{xy,zk\notin \Y\mid e\in \Y}=1-\frac{4}{9}=\frac{5}{9}$.
Hence, 
\begin{align*}
\PP{e\in \Z\mid e\in \Y}=&\ \PP{e\in \Z\mid xy,zk\notin \Y \wedge  e\in \Y}\cdot\PP{xy,zk\notin \Y\mid e\in \Y}\\
&\  +\PP{e\in \Z\mid (xy\in \Y \vee zk\in \Y) \wedge  e\in \Y}\cdot\PP{xy\in \Y \vee zk\in \Y\mid e\in \Y}\\
=&\ 1\cdot\frac{4}{9}+\frac{1}{2}\cdot\frac{5}{9}\\
=&\ \frac{13}{18}.
\end{align*}

\textbf{Case~5: both $x$ and $z$ are external vertices, there exists a non external vertex in $C$, $\size{C}>3$, and $\size{C}\bmod2=1$.}
Note that the component containing $e$ is an odd cycle only if $C$ is an odd cycle and all vertices of $C$ are external vertices. Hence, the component containing $e$ cannot be an odd cycle in this case.
The analysis is the same as in Case 4.
We have $\PP{e\in \Z\mid e\in \Y}=\frac{13}{18}$.

\textbf{Case~6: all vertices of $C$ are external vertices, $\size{C}>3$, and $\size{C}\bmod2=1$.} Let $\size{C}=l$.
Each vertex of $C$ corresponds to a different triangle, and then we have $\PP{C\in \Y\mid e\in \Y}=(\frac{1}{3})^{l-2}$, $\PP{xy,zk\notin \Y\mid e\in \Y}=\frac{2}{3}\cdot\frac{2}{3}=\frac{4}{9}$, and $\PP{(xy\in \Y\vee zk\in \Y)\wedge C\notin \Y \mid e\in \Y}=1-(\frac{1}{3})^{l-2}-\frac{4}{9}$.
If $C\in \Y$ (i.e., all edges of $C$ are in $\Y$), the component containing $e$ is an odd cycle, and then we have $\PP{e\in \Z\mid C\in \Y}=\frac{l-1}{l}\cdot\frac{1}{2}$.
If $e\in \Y$ and $xy,zk\notin \Y$, the component containing $e$ is an edge, and then we have $\PP{e\in \Z\mid xy,zk\notin \Y\wedge e\in \Y }=1$.
If $e\in \Y$ and $(xy\in \Y\vee zk\in \Y)\wedge C\notin \Y$, the component containing $e$ is a path, and then we have $\PP{e\in \Z\mid (xy\in \Y\vee zk\in \Y)\wedge C\notin \Y }=\frac{1}{2}$.
Hence, 
\begin{align*}
&\PP{e\in \Z\mid e\in \Y}\\
&=\PP{e\in \Z\mid C\in \Y}\cdot\PP{C\in \Y\mid e\in \Y}\\
&\ \ \  +\PP{e\in \Z\mid xy,zk\notin \Y \wedge  e\in \Y}\cdot\PP{xy,zk\notin \Y\mid e\in \Y}\\
&\ \ \ +\PP{e\in \Z\mid (xy\in \Y\vee zk\in \Y)\wedge C\notin \Y \wedge  e\in \Y}\cdot\PP{(xy\in \Y\vee zk\in \Y)\wedge C\notin \Y\mid e\in \Y}\\
&=\ \frac{l-1}{l}\cdot\frac{1}{2}\cdot\lrA{\frac{1}{3}}^{l-2}+1\cdot\frac{4}{9}+\frac{1}{2}\cdot\lrA{1-\lrA{\frac{1}{3}}^{l-2}-\frac{4}{9}}\\
&=\ \frac{13}{18}-\frac{1}{2l}\cdot\lrA{\frac{1}{3}}^{l-2}\\
&\geq\ \frac{97}{135},
\end{align*}
where we have $l=5$ in the worst case.
\end{proof}

Note that the idea of the randomized matching algorithm in~\cite{DBLP:journals/jco/ChenCLWZ21} is simply to decompose the out-edges in $\Y$ into three matchings and put one into $\Z$, i.e., $\PP{e\in \Z\mid e\in \Y}=\frac{1}{3}$, and then we may only get a probability of $\frac{1}{27}$ for $\PP{e\in \Z}$, which is smaller than ours in Lemma~\ref{p}. Their decomposition is simple because in their algorithm the out-edges may form a multi-graph where parallel edges exist. As a result, the structure is very complicated, and it is not even obvious to get a better decomposition.

\begin{lemma}
It holds that
\begin{align*}
\EE{w(\Z)}>&\ \frac{1}{3}w(\B^*)+\frac{97(1-3\tau)}{7290}w(\B^*_{2}) +\frac{97(1-3\tau)}{2430}w(\B^*_{4}).
\end{align*}
\end{lemma}
\begin{proof}
Recall that $0\leq \tau\leq 1/3$ by the definition. Hence, we have $\frac{1-3\tau}{6}\geq 0$. By Lemmas~\ref{e} and \ref{p}, we have
\begin{align*}
\EE{w(\Z)}>&\ \frac{1}{3}w(\B^*)+\sum_{t\in\B^*_{2},e\in E_t}\lrA{\frac{1-3\tau}{6}w(t)\cdot\PP{e\in\Z}} +\sum_{t\in\B^*_{4}}\lrA{\frac{1-3\tau}{6}w(t)\cdot\sum_{e\in E_t}\PP{e\in\Z}}\\
\geq&\ \frac{1}{3}w(\B^*)+\sum_{t\in\B^*_{2}}\lrA{\frac{1-3\tau}{6}w(t)\cdot\frac{97}{1215}} +\sum_{t\in\B^*_{4}}\lrA{\frac{1-3\tau}{6}w(t)\cdot3\cdot\frac{97}{1215}}\\
=&\ \frac{1}{3}w(\B^*)+\frac{97(1-3\tau)}{7290}w(\B^*_{2}) +\frac{97(1-3\tau)}{2430}w(\B^*_{4}).
\end{align*}
\end{proof}

Recall that $w(\T_2)\geq 2w(\M^*)$ and $w(\M^*)\geq \EE{w(\Z)}$. Then, we have the following lemma.
\begin{lemma}\label{t2+}
The triangle packing $\T_2$ satisfies that
\begin{align*}
w(\T_2)\geq &\ \frac{2}{3}w(\B^*)+\frac{97(1-3\tau)}{3645}w(\B^*_{2}) +\frac{97(1-3\tau)}{1215}w(\B^*_{4}).
\end{align*}
\end{lemma}

\section{The Third Triangle Packing}\label{third}
Recall that we split triangles into type-1 and type-2 triangles. Especially, the randomized matching in Section~\ref{second} is based on type-1 triangles. However, there may not exist type-1 triangles in the worst case. Then, it comes to the third triangle packing $\T_3$, which uses the property of type-2 triangles. 
The ideas of $\T_3$ are mainly from the randomized algorithm in~\cite{DBLP:journals/jco/ChenCLWZ21}. We will first propose a new randomized algorithm and then derandomize it.


\subsection{The randomized algorithm}
Our new randomized algorithm of $\T_3$ is based on two edge-disjoint matchings $\X$ and $\Y$, where $\X$ is a randomized matching of size $n/3$ on the short cycle packing and $\Y$ is a matching of size at most $n/3$ determined by $\X$. The algorithm of $\X$ in our algorithm is obtained directly from~\cite{DBLP:journals/jco/ChenCLWZ21} but the algorithm of $\Y$ is different: the triangle packing algorithm in~\cite{DBLP:journals/jco/ChenCLWZ21} is randomized even if $\X$ is determined, while our algorithm of $\T_3$ would be deterministic. So, our new randomized algorithm can be derandomized simply by finding a desirable deterministic matching $\X$. 

The algorithm of $\X$ in~\cite{DBLP:journals/jco/ChenCLWZ21} mainly contains four steps.

\medskip
\noindent\textbf{Step~1.} Initialize $\L=\R=\emptyset$.

\noindent\textbf{Step~2.} For each even cycle $C\in\C$, partition it into two matchings, select one matching uniformly at random, and select the edges in the chosen matching into $\L$.

\noindent\textbf{Step~3.} For each odd cycle $C\in\C$, select one edge uniformly at random, delete it, partition the rest edges into two matchings, select one matching uniformly at random, and select the edges in the chosen matching into $\L$ and $\R$ in the following way:
select one edge uniformly at random, put the edge into $\R$, and put the rest edges into $\L$.

\noindent\textbf{Step~4.} Select $\frac{2}{3}\size{\L}$ edges of $\L$ uniformly at random into $\X$, and put all edges of $\R$ into $\X$.
\medskip

The running time is dominated by computing the short cycle packing $\C$, which is polynomial. The randomized matching $\X$ has some good properties. We have the following lemma.

\begin{lemma}[\cite{DBLP:journals/jco/ChenCLWZ21}]~\label{p+}
The size of $\X$ is $n/3$. For every edge $e\in C\in\C$ and every vertex $v\notin C$, we have $\PP{e\in\X}=\frac{1}{3}$ and $\PP{e\in\X\wedge v\notin \X}\geq\frac{1}{9}$.
\end{lemma}

We remark that if $\C$ is a set of triangles the matching $\X$ is exactly obtained by taking an edge of each triangle uniformly at random. In this case, for every edge $e\in C\in\C$ and every vertex $v\notin C$, we have $\PP{e\in\X}=\frac{1}{3}$ and $\PP{e\in\X\wedge v\notin \X}=\frac{1}{9}$. So, the result in Lemma \ref{p+} is essentially tight.

Next, we introduce the algorithm of the third triangle packing $\T_3$, where we will first compute the matching $\Y$ and then obtain $\T_3$ using $\X$ and $\Y$. We first give some definitions.

For some edge $e=xy\in C\in\C$ and vertex $z\notin C$, we call $(x,y;z)$ a \emph{triplet}. Recall that $\tau$ is a fixed constant defined before. A triplet $(x,y;z)$ is \emph{good} if it satisfies that $w(xy)\leq (1-\tau)(w(xz)+w(yz))$.

We construct a multi-graph $H$. Initially, graph $H$ contains $n$ vertices only and no edges. For each good triplet $(x,y;z)$, we add two edges $xz$ and $yz$ with an augmented weight $\widetilde{w}(xz)=\widetilde{w}(yz)=w(xz)+w(yz)$, and each of them has a label corresponding to the good triplet $(x,y;z)$. Hence, graph $H$ contains only external edges. Two edges of $H$ are called \emph{conflicting} if their corresponding triplets share a common vertex.

The algorithm of $\T_3$ is shown as follows.

\medskip
\noindent\textbf{Step~1.} Find a maximum augmented weight matching $\Y^{*}$ in graph $H$ in $O(n^3)$ time~\cite{gabow1974implementation,lawler1976combinatorial}. Note that there is no constraint on its size, and hence we have $0\leq\size{\Y^{*}}\leq n/2$.

\noindent\textbf{Step~2.} Let $\Y^*_\X$ denote the set of edges $\{zx\in\Y^*\mid xy\in\X,z\notin\X\}$. 
We claim that for each edge in $\Y^*_\X$ there is at most one different edge in $\Y^*_\X$ that are conflicting with it (see Lemma~\ref{conflict}).

\noindent\textbf{Step~3.} For each pair of conflicting edges of $\Y^*_\X$, delete the edge with the less augmented weight (if their augmented weights are the same we may simply delete one of them arbitrarily). The remained matching is denoted by $\Y$. 

\noindent\textbf{Step~4.} Note that $\X\cup\Y$ is a set of components of size $n/3$, which contains $\size{\Y}$ path-components (each contains 3 vertices) and $\size{\X}-\size{\Y}$ edge-components. For each path-component $xyz$, complete it into a triangle. For each edge-component, arbitrarily assign an unused vertex, and complete it into a triangle.
\medskip


We need to prove the claim in {Step~2}.

\begin{lemma}\label{conflict}
For any edge $zx\in \Y^*_\X$, there is at most one different edge in $\Y^*_\X$ that are conflicting with it.
\end{lemma}
\begin{proof}
Consider an edge $zx\in\Y^*_\X$ with its corresponding triplet $(x,y;z)$. Assume that $z'x'\in\Y^*_\X$ with its corresponding triplet $(x',y';z')$ is conflicting with $zx$, i.e., $\{x,y,z\}\cap\{x',y',z'\}\neq\emptyset$. Since $\Y^*_\X$ is a matching, we get that $zx$ and $z'x'$ are vertex-disjoint. Moreover, we have $xy,x'y'\in\X$ and $z,z'\notin\X$ by the definition of $\Y^*_\X$. So, we have $\{x,y\}\cap\{x',y'\}\neq\emptyset$. 
Since $\X$ is a matching, we must have $x'=y$ and $x=y'$. So, $z'x'=z'y$. Since $\Y^*_\X$ is a matching, there is at most one such edge.
\end{proof}

\subsection{The analysis}
\begin{lemma}\label{compute}
$\EE{w(\T_3)}\geq \frac{\tau}{18}\widetilde{w}(\Y^*)+\frac{2}{3}w(\C)$.
\end{lemma}
\begin{proof}
We use $\X_\Y$ to denote the edges of $\X$ contained in the path-components, and hence $\X\setminus\X_\Y$ denote the edges in the edge-components. For the triangles corresponding to the path-components, the weight is exactly $\widetilde{w}(\Y)+w(\X_\Y)$. Note that each triangle is also a good triplet. Hence, we can get $w(\X_\Y)\leq (1-\tau)\widetilde{w}(\Y)$ and in turn $\widetilde{w}(\Y)-w(\X_\Y)\geq \tau\widetilde{w}(\Y)$.
For the triangles corresponding to the edge-components, the weight is at least $2w(\X\setminus\X_\Y)$ by the triangle inequality.
Hence, we have 
\[
w(\T_3)\geq \widetilde{w}(\Y)+w(\X_\Y)+2w(\X\setminus\X_\Y)=\widetilde{w}(\Y)-w(\X_\Y)+2w(\X)\geq \tau\widetilde{w}(\Y)+2w(\X).
\]
Recall that for each pair of conflicting edges we delete the edge with the less augmented weight. Hence, we have $\widetilde{w}(\Y)\geq\frac{1}{2}\widetilde{w}(\Y^*_\X)$, and then 
\[
w(\T_3)\geq \frac{\tau}{2}\widetilde{w}(\Y^*_\X)+2w(\X).
\]

Then, we consider $\EE{w(\X)}$ and $\EE{\widetilde{w}(\Y^*_\X)}$.
By Lemma~\ref{p+}, we have 
\[
\EE{w(\X)}=\sum_{e\in\C}w(e)\cdot\PP{e\in\X}=\frac{1}{3}w(\C)
\]
and
\[
\EE{\widetilde{w}(\Y^*_\X)}=\sum_{\substack{zx\in \Y^*:\\(x,y;z)}}\widetilde{w}(zx)\cdot\PP{xy\in \X\wedge z\notin\X}\geq \sum_{\substack{zx\in \Y^*:\\(x,y;z)}}\frac{1}{9}\widetilde{w}(zx)=\frac{1}{9}\widetilde{w}(\Y^*).
\]
Hence, $\EE{w(\T_3)}\geq \frac{\tau}{18}\widetilde{w}(\Y^*)+\frac{2}{3}w(\C)$.
\end{proof}

Next, we give a lower bound of $\widetilde{w}(\Y^*)$.

\begin{lemma}
It holds that
\[
\widetilde{w}(\Y^*)\geq \sum_{i\in\{3,5\}}\frac{1}{2}w(\B^*_{i}).
\]

\end{lemma}
\begin{proof}
We will use the edges of type-2 triangles to construct a matching of graph $H$. First, we will show that for each type-2 triangle $t=xyz$ it is incident to a good triplet.

By the definition, there is one edge w.l.o.g. $e=xx'\in E_t$ such that $w(xx')\leq\frac{1}{2}(1-\tau)w(t)$, where $x$ is an external vertex. Note that vertex $y$ or $z$ cannot be on the same cycle with vertex $x$. Hence, there are two triplets $(x,x';y)$ and $(x,x';z)$. We will show that there is a good triplet between them. See Figure~\ref{fig04} for an illustration.

\begin{figure}[ht]
\centering
\begin{tikzpicture}
\tikzstyle{vertex}=[black,circle,fill,minimum size=5,inner sep=0pt]
\begin{scope}[every node/.style={vertex}]
 \node (x) at (0,0) {};
 \node (z) at (-0.5,1.2) {};
 \node (y) at (0.5,1.2) {};
 \node (x'') at (-1,0) {};
 \node (x') at (1,0) {};
\end{scope}
\node at (0,0.8) {$t$};
\node[above] at (0,-0.5) {$x$};
\node[above] at (1,-0.5) {$x'$};
\node[above] at (y) {$y$};\node[above] at (z) {$z$};
\draw[very thick,-] (z.west) to (y.east);
\draw[very thick,-] (x.north) to (y.south);
\draw[very thick,-] (x.north) to (z.south);
\draw[very thick,dotted,->] (x''.east) to (x.west);
\draw[very thick,dotted,->,red] (x.east) to (x'.west);
\end{tikzpicture}
\caption{An illustration of the good triplet of a type-2 triangle $t=xyz$, where the directed dotted edges are the edges of $\C$ and the red directed dotted edge $xx'$ is an out-edge in $E_t$ such that $w(xx')\leq\frac{1}{2}(1-\tau)w(t)$}
\label{fig04}
\end{figure}
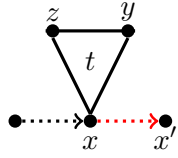

Note that $w(x'y)+w(x'z)\geq w(yz)$ by the triangle inequality. We have 
$
w(xy)+w(x'y)+w(xz)+w(x'z)\geq w(xy)+w(xz)+w(yz)=w(t).
$
Hence, we have $\max\{w(xy)+w(x'y),w(xz)+w(x'z)\}\geq \frac{1}{2}w(t)$. Assume w.l.o.g. that $w(xy)+w(x'y)\geq \frac{1}{2}w(t)$. Since $w(xx')\leq\frac{1}{2}(1-\tau)w(t)$, we have $w(xx')\leq (1-\tau)(w(xy)+w(x'y))$. Hence, $(x,x';y)$ is a good triplet.

According to the good triplet, there is an edge $xy$ in graph $H$ such that $\widetilde{w}(xy)=w(xy)+w(x'y)\geq \frac{1}{2}w(t)$.
By taking the edge $xy$ for each type-2 triangle, we can get a matching of graph $H$ with a weight of at least
\[
\sum_{t\in\B^*_3\cup\B^*_5}\frac{1}{2}w(t)=\sum_{i\in\{3,5\}}\frac{1}{2}w(\B^*_{i}).
\]
Since $\Y^*$ is the maximum weight matching in graph $H$, we can get the lower bound.
\end{proof}

Recall that $w(\C)\geq (1-\varepsilon)w(\C^*)\geq(1-\varepsilon)w(\B^*)$ and $\EE{w(\T_3)}\geq\frac{\tau}{18}\widetilde{w}(\Y^*)+\frac{2}{3}w(\C)$. Then, we have the following lemma.
\begin{lemma}\label{t3}
There is a polynomial-time randomized algorithm that can compute a triangle packing $\T_3$ such that
\[
\EE{w(\T_3)}\geq \frac{2}{3}(1-\varepsilon)w(\B^*)+\sum_{i\in\{3,5\}}\frac{\tau}{36}w(\B^*_{i}).
\]
\end{lemma}

Next, we show that the algorithm can be derandomized efficiently by the method of conditional expectations~\cite{williamson2011design}.

\subsection{The derandomization}
The third triangle packing $\T_3$ is randomized due to the randomized matching $\X$ on $\C$. Let $f(\X)=\frac{\tau}{2}\widetilde{w}(\Y^*_\X)+2w(\X)$. By the proof of Lemma~\ref{compute}, we have $w(\T_3)\geq f(\X)$ and $\EE{f(\X)}\geq \frac{\tau}{18}\widetilde{w}(\Y^*)+\frac{2}{3}w(\C)$. Therefore, if we can derandomize the matching $\X$ such that $f(\X)\geq \frac{\tau}{18}\widetilde{w}(\Y^*)+\frac{2}{3}w(\C)$, then we will obtain a deterministic algorithm of $\T_3$. 

Recall that the randomized matching algorithm mainly contains two phases. In the first phase, the algorithm obtains two edge sets $\L$ and $\R$ by making random decisions on each cycle in $\C$. After that, the algorithm obtains the matching $\X$ by choosing $\frac{2}{3}\size{\L}$ edges in $\L$ uniformly at random and all edges in $\R$, which can be seen as making random decisions on each edge in $\L$.

Using the method of conditional expectations, these two phases can be derandomized efficiently. The idea is to consider random decisions in the algorithm sequentially. For each random decision, we can compute the expected weight of $f(\X)$ conditioned on each possible outcome of the random decision. Since at least one of these outcomes has an expected weight of at least $\EE{f(\X)}$, we fix this outcome and continue to the next random decision. By repeating this procedure, we can get a deterministic solution with a weight of at least $\EE{f(\X)}$.

We will show that there is a polynomial number of outcomes for each random decision in the algorithm. Moreover, if we can compute the conditional expected weight of $f(\X)$ in polynomial time, then the derandomization will take polynomial time. Since $f(\X)=\frac{\tau}{2}\widetilde{w}(\Y^*_\X)+2w(\X)$, by the linearity, it is sufficient to show that we can compute the conditional expected weight of $\widetilde{w}(\Y^*_\X)$ and $w(\X)$ in polynomial time, respectively. 


Next, we consider the derandomization of these two phases, respectively.

\begin{lemma}\label{derandom1}
The first phase of the algorithm can be derandomized in polynomial time.
\end{lemma}
\begin{proof}
We consider random decisions on each cycle in $\C$ sequentially.
Let $\C=\{C_1,\dots,C_{\size{\C}}\}$. For the decision on $C_i\in\C$, it contains $2$ outcomes if $C_i$ is an even cycle, and $\size{C_i}\cdot(\size{C_i}-1)=O_\varepsilon(1)$ outcomes otherwise. 
Next, we show that we can use polynomial time to compute the conditional expected weight of $w(\X)$ and $\widetilde{w}(\Y^*_\X)$. For the sake of presentation, for some variable $\Z$, we simply use $\EE{\Z\mid C_1,\dots,C_s}$ to denote the expected weight of $\Z$ conditioned on the $s$ outcomes of the first $s$ decisions on $C_1,\dots,C_s$. 

\textbf{The conditional expected weight of $w(\X)$.} Since $\X$ is a matching on $\C$, we have
\[
\EE{w(\X)\mid C_1,\dots,C_s}=\sum_{C_i\in\C}\sum_{e\in C_i}\lrA{w(e)\cdot\PP{e\in\X\mid C_1,\dots,C_s}}.
\]
Assume that $e\in C_i\in\C$. 
Since the algorithm selects $\frac{2}{3}\size{\L}$ edges in $\L$ uniformly at random and all edges in $\R$ into $\X$, we can get 
\begin{align*}
&\PP{e\in\X\mid C_1,\dots,C_s}\\
&\ =\frac{2}{3}\cdot\PP{e\in\L\mid C_1,\dots,C_s}+1\cdot\PP{e\in\R\mid C_1,\dots,C_s}+0\cdot\PP{e\notin\L\wedge e\notin\R\mid C_1,\dots,C_s}.
\end{align*}
Hence, it is sufficient to compute the conditional probabilities of events $e\in\L$, $e\in\R$, and $e\notin\L\wedge e\notin\R$. If $i\leq s$, we can determine $\PP{e\in\L\mid C_1,\dots,C_s}$, $\PP{e\in\R\mid C_1,\dots,C_s}$, and $\PP{e\notin\L\wedge e\notin\R\mid C_1,\dots,C_s}$ based on the outcome of the decision on $C_i$. Otherwise, let $k=\size{C_i}$, and by the algorithm, we can compute them as follows:
\[
\left\{
\begin{aligned}
&\PP{e\in\L\mid C_1,\dots,C_s}=\frac{1}{2},\quad k\bmod 2=0,\\ 
&\PP{e\in\R\mid C_1,\dots,C_s}=0,\quad k\bmod 2=0,\\
&\PP{e\notin\L\wedge e\notin\R\mid C_1,\dots,C_s}=\frac{1}{2},\quad k\bmod 2=0,\\
&\PP{e\in\L\mid C_1,\dots,C_s}=\frac{k-3}{2k},\quad k\bmod 2=1,\\
&\PP{e\in\R\mid C_1,\dots,C_s}=\frac{1}{k},\quad k\bmod 2=1,\\
&\PP{e\notin\L\wedge e\notin\R\mid C_1,\dots,C_s}=\frac{k+1}{2k},\quad k\bmod 2=1.
\end{aligned}
\right.
\]
Hence, the conditional expected weight of $w(\X)$ can be computed in polynomial time.

\textbf{The conditional expected weight of $\widetilde{w}(\Y^*_\X)$.} Since $\Y^*_\X\subseteq\Y^*$, we have
\[
\EE{\widetilde{w}(\Y^*_\X)\mid C_1,\dots,C_s}=\sum_{\substack{zx\in \Y^*:\\(x,y;z)}}\lrA{(w(zx)+w(zy))\cdot\PP{xy\in\X\wedge z\notin\X\mid C_1,\dots,C_s}}.
\]
Assume that $xy\in C_i$ and $z\in C_j$, where $i\neq j$ by the definition of the good triplet $(x,y;z)$. There are nine events to consider based on whether the two disjoint sets $\L$ and $\R$ contain $xy$ or $z$. Since the algorithm selects $\frac{2}{3}\size{\L}$ edges in $\L$ uniformly at random and all edges in $\R$ into $\X$, we can get
\begin{align*}
&\PP{xy\in\X\wedge z\notin\X\mid C_1,\dots,C_s}\\
&=\ 0\cdot\PP{xy\notin\L\wedge xy\notin\R\mid C_1,\dots,C_s}+0\cdot\PP{z\in\R\mid C_1,\dots,C_s}\\
&\ +\frac{\binom{\size{\L}-2}{\frac{2}{3}\size{\L}-1}}{\binom{\size{\L}}{\frac{2}{3}\size{\L}}}\cdot\PP{xy\in\L\wedge z\in\L\mid C_1,\dots,C_s}+\frac{2}{3}\cdot\PP{xy\in\L\wedge z\notin\L\wedge z\notin\R\mid C_1,\dots,C_s}\\
&\ +\frac{1}{3}\cdot\PP{xy\in\R\wedge z\in\L\mid C_1,\dots,C_s}+1\cdot\PP{xy\in\R\wedge z\notin\L\wedge z\notin\R\mid C_1,\dots,C_s},
\end{align*}
where $\size{\L}$ is a number related to the numbers of odd cycles in $\C$, which can be computed as follows.
Let $\C_o$ and $\C_e$ be the set of odd cycles and even cycles in $\C$, respectively. By the algorithm, we can get $\size{\L}=\sum_{C \in\C_o}(\frac{1}{2}(\size{C}-1)-1)+\sum_{C \in\C_e}\frac{1}{2}\size{C}=\frac{1}{2}(n-3\size{\C_o})$. The event on $xy$ is independent with the event on $z$ since in the first phase the algorithm makes random decisions on each cycle independently.
Therefore, it is sufficient to compute the probabilities of events $xy\in\L$, $xy\in\R$, $xy\notin\L\wedge xy\notin R$, $z\in\L$, $z\in\R$, and $z\notin\L\wedge z\notin R$ conditioned on $C_1,\dots,C_s$. 
Using a similar argument as in the previous case, we can compute the conditional probabilities of events $xy\in\L$, $xy\in\R$, and $xy\notin\L\wedge xy\notin R$ in polynomial time.
Next, we compute the conditional probabilities of events $z\in\L$, $z\in\R$, and $z\notin\L\wedge z\notin R$. Similarly, if $j\leq s$, we can determine $\PP{z\in\L\mid C_1,\dots,C_s}$, $\PP{z\in\R\mid C_1,\dots,C_s}$, and $\PP{z\notin\L\wedge z\notin\R\mid C_1,\dots,C_s}$ based on the outcome of the decision on $C_j$. Otherwise, let $k=\size{C_j}$, and by the algorithm, we can compute them as follows:
\[
\left\{
\begin{aligned}
&\PP{z\in\L\mid C_1,\dots,C_s}=1,\quad k\bmod 2=0,\\ 
&\PP{z\in\R\mid C_1,\dots,C_s}=0,\quad k\bmod 2=0,\\
&\PP{z\notin\L\wedge z\notin\R\mid C_1,\dots,C_s}=0,\quad k\bmod 2=0,\\
&\PP{z\in\L\mid C_1,\dots,C_s}=\frac{k-3}{k},\quad k\bmod 2=1,\\
&\PP{z\in\R\mid C_1,\dots,C_s}=\frac{2}{k},\quad k\bmod 2=1,\\
&\PP{z\notin\L\wedge z\notin\R\mid C_1,\dots,C_s}=\frac{1}{k},\quad k\bmod 2=1.
\end{aligned}
\right.
\]
Therefore, the conditional expected weight of $\widetilde{w}(\Y^*_\X)$ can also be computed in polynomial time.
\end{proof}

\begin{lemma}\label{derandom2}
The second phase of the algorithm can be derandomized in polynomial time.
\end{lemma}
\begin{proof}
Currently, we have determined the edges in $\L$ and $\R$ using the derandomization of the first phrase. Let $\L=\{e_1,\dots,e_{\size{\L}}\}$. Recall that the matching $\X$ is obtained by selecting $\frac{2}{3}\size{\L}$ edges in $\L$ uniformly at random and all edges in $\R$. 
We consider random decisions on each edge in $\L$ sequentially.
Clearly, for each edge in $\L$ there are only 2 outcomes according to whether it is selected into $\X$ or not. Next, we show that we can use polynomial time to compute the conditional expected weight of $w(\X)$ and $\widetilde{w}(\Y^*_\X)$. We use $b_i=1$ to denote $e_i\in\X$ and $b_i=0$ to denote $e_i\notin\X$. For some variable $\Z$, we use $\EE{\Z\mid b_1,\dots,b_s}$ to denote the expected weight of $\Z$ conditioned on the $s$ outcomes of the first $s$ random decisions on $e_1,\dots,e_s$. Note that the algorithm will select $\frac{2}{3}\size{\L}-b_1-\dots-b_s$ edges from $\L\setminus\{e_1,\dots,e_s\}$ uniformly at random, where $b_1+\dots+b_s\leq \frac{2}{3}\size{\L}$. Once we have determined $\frac{2}{3}\size{\L}$ edges that can be selected into $\X$, i.e., $b_1+\dots+b_s=\frac{2}{3}\size{\L}$ for some $s$, then we are done. Currently, we may assume w.l.o.g. that $b_1+\dots+b_s<\frac{2}{3}\size{\L}$.

\textbf{The conditional expected weight of $w(\X)$.} We have
\[
\EE{w(\X)\mid b_1,\dots,b_s}=\sum_{e\in \L}\lrA{w(e)\cdot\PP{e\in\X\mid b_1,\dots,b_s}}+\sum_{e\in \R}\lrA{w(e)\cdot\PP{e\in\X\mid b_1,\dots,b_s}}.
\]
Consider any edge $e_i\in\L$. If $i\leq s$, the conditional probability of $e_i\in\X$ equals to $b_i$ based on the outcome of the decision on $e_i$. Otherwise, the conditional probability is $(\frac{2}{3}\size{\L}-b_1-\dots-b_s)/(\size{\L}-s)$. Note that for any edge $e\in \R$ the conditional probability of $e\in\X$ is 1. Hence, the conditional expected weight of $w(\X)$ can be computed in polynomial time.

\textbf{The conditional expected weight of $\widetilde{w}(\Y^*_\X)$.} We have
\[
\EE{\widetilde{w}(\Y^*_\X)\mid b_1,\dots,b_s}=\sum_{\substack{zx\in \Y^*:\\(x,y;z)}}\lrA{(w(zx)+w(zy))\cdot\PP{xy\in\X\wedge z\notin\X\mid b_1,\dots,b_s}}.
\]
Similar to the derandomization of the first phrase, there are nine cases (note that we use the term `cases' rather than `events' because in the second phrase the edges in $\L$ and $\R$ have been determined) to consider based on whether the two disjoint sets $\L$ and $\R$ contain $xy$ or $z$. Moreover, if $xy\notin\L\wedge xy\notin\R$ or $z\in\R$, the conditional probability of $xy\in\X\wedge z\notin\X$ is 0. Hence, there are still four cases: $xy\in\L\wedge z\in\L$, $xy\in\L\wedge z\notin\L\wedge z\notin\R$, $xy\in\R\wedge z\in\L$, and $xy\in\R\wedge z\notin\L\wedge z\notin\R$. The decisions on the edges $e_1,\dots,e_s$ in $\L$ have been determined based on the conditions, and the algorithm will select $\frac{2}{3}\size{\L}-b_1-\dots-b_s$ edges from $\L\setminus\{e_1,\dots,e_s\}$ uniformly at random into $\X$. Hence, we need to further consider sub-cases based on whether $\{e_1,\dots,e_s\}$ contains $xy$ or $z$. 

\noindent\textbf{Case~1: $xy\in\L\wedge z\in\L$.} There are four sub-cases based on whether $\{e_1,\dots,e_s\}$ contains $xy$ or $z$.

\noindent\textbf{Case~1.1: $xy,z\notin\{e_1,\dots,e_s\}$.} Since the algorithm selects $\frac{2}{3}\size{\L}-b_1-\dots-b_s$ edges from $\L\setminus\{e_1,\dots,e_s\}$ uniformly at random into $\X$, we can get
\[
\PP{xy\in\X\wedge z\notin\X\mid b_1,\dots,b_s}=\frac{\binom{\size{\L}-s-2}{\frac{2}{3}\size{\L}-b_1-\dots-b_s-1}}{\binom{\size{\L}-s}{\frac{2}{3}\size{\L}-b_1-\dots-b_s}}.
\]
\noindent\textbf{Case~1.2: $xy\notin\{e_1,\dots,e_s\}$ and $z\in\{e_1,\dots,e_s\}$.} Assume that $z\in e_j\in\{e_1,\dots,e_s\}$. If $b_j=1$, we have $z\in \X$, and then $\PP{xy\in\X\wedge z\notin\X\mid b_1,\dots,b_s}=0$. Otherwise, we can get 
\[
\PP{xy\in\X\wedge z\notin\X\mid b_1,\dots,b_s}=\frac{\binom{\size{\L}-s-1}{\frac{2}{3}\size{\L}-b_1-\dots-b_s-1}}{\binom{\size{\L}-s}{\frac{2}{3}\size{\L}-b_1-\dots-b_s}}.
\]
\noindent\textbf{Case~1.3: $xy\in\{e_1,\dots,e_s\}$ and $z\notin\{e_1,\dots,e_s\}$.} Assume that $xy=e_i\in\{e_1,\dots,e_s\}$. If $b_i=0$, we have $xy\notin \X$, and then $\PP{xy\in\X\wedge z\notin\X\mid b_1,\dots,b_s}=0$. Otherwise, we can get 
\[
\PP{xy\in\X\wedge z\notin\X\mid b_1,\dots,b_s}=\frac{\binom{\size{\L}-s-1}{\frac{2}{3}\size{\L}-b_1-\dots-b_s}}{\binom{\size{\L}-s}{\frac{2}{3}\size{\L}-b_1-\dots-b_s}}.
\]

\noindent\textbf{Case~1.4: $xy,z\in\{e_1,\dots,e_s\}$.} Assume that $xy=e_i$ and $z\in e_j$ such that $e_i,e_j\in\{e_1,\dots,e_s\}$. We know that $\PP{xy\in\X\wedge z\notin\X\mid b_1,\dots,b_s}=1$ if and only if $b_i=1$ and $b_j=0$.

\noindent\textbf{Case~2: $xy\in\L\wedge z\notin\L\wedge z\notin\R$.} We get two sub-cases based on whether $\{e_1,\dots,e_s\}$ contains $xy$.

\noindent\textbf{Case~2.1: $xy\notin\{e_1,\dots,e_s\}$.} Since $z\notin\L\wedge z\notin\R$, we have $z\notin\X$. Moreover, since the algorithm selects $\frac{2}{3}\size{\L}-b_1-\dots-b_s$ edges from $\L\setminus\{e_1,\dots,e_s\}$ uniformly at random into $\X$, we can get
\[
\PP{xy\in\X\wedge z\notin\X\mid b_1,\dots,b_s}=\frac{\binom{\size{\L}-s-1}{\frac{2}{3}\size{\L}-b_1-\dots-b_s-1}}{\binom{\size{\L}-s}{\frac{2}{3}\size{\L}-b_1-\dots-b_s}}.
\]
\noindent\textbf{Case~2.2: $xy\in\{e_1,\dots,e_s\}$.} Assume that $xy=e_i\in\{e_1,\dots,e_s\}$. Since $z\notin\X$, we know that $\PP{xy\in\X\wedge z\notin\X\mid b_1,\dots,b_s}=1$ if and only if $b_i=1$.

\noindent\textbf{Case~3: $xy\in\R\wedge z\in\L$.} We get two sub-cases based on whether $\{e_1,\dots,e_s\}$ contains $z$.

\noindent\textbf{Case~3.1: $z\notin\{e_1,\dots,e_s\}$.} Since $xy\in\R$, we have $xy\in\X$. Moreover, since the algorithm selects $\frac{2}{3}\size{\L}-b_1-\dots-b_s$ edges from $\L\setminus\{e_1,\dots,e_s\}$ uniformly at random into $\X$, we can get
\[
\PP{xy\in\X\wedge z\notin\X\mid b_1,\dots,b_s}=\frac{\binom{\size{\L}-s-1}{\frac{2}{3}\size{\L}-b_1-\dots-b_s}}{\binom{\size{\L}-s}{\frac{2}{3}\size{\L}-b_1-\dots-b_s}}.
\]
\noindent\textbf{Case~3.2: $z\in\{e_1,\dots,e_s\}$.} Assume that $z\in e_j\in\{e_1,\dots,e_s\}$. Since $xy\in\X$, we know that $\PP{xy\in\X\wedge z\notin\X\mid b_1,\dots,b_s}=1$ if and only if $b_j=0$.

\noindent\textbf{Case~4: $xy\in\R\wedge z\notin\L\wedge z\notin\R$.} Since $xy\in\R$, we get $xy\in\X$. Moreover, since $z\notin\L\wedge z\notin\R$, we get $z\notin\X$. Therefore, we get $\PP{xy\in\X\wedge z\notin\X\mid b_1,\dots,b_s}=1$.


Hence, the conditional expected weight of $\widetilde{w}(\Y^*_\X)$ can be computed in polynomial time.
\end{proof}

By Lemmas~\ref{derandom1} and \ref{derandom2}, we have the following lemma.
\begin{lemma}\label{t3+}
There is a polynomial-time algorithm that can compute a triangle packing $\T_3$ such that
\[
w(\T_3)\geq \frac{2}{3}(1-\varepsilon)w(\B^*)+\sum_{i\in\{3,5\}}\frac{\tau}{36}w(\B^*_{i}).
\]
\end{lemma}

\section{The Trade-off}\label{analysis}
We are ready to make a trade-off among the three triangle packings: $\T_1$, $\T_2$, and $\T_3$.

We first introduce some new parameters.
We define $\alpha_i=w(\B^*_i)/w(\B^*)$, which measures the weight proportion of the triangles in $\B^*_i$ compared to the triangles in $\B^*$. Note that
\begin{equation}\label{lp1}
\alpha_1+\alpha_2+\alpha_3+\alpha_4+\alpha_5=1.
\end{equation}
Then, we define $\rho_i=\sum_{t\in\B^*_i}w(a_t)/w(\B^*)$, $\sigma_i=\sum_{t\in\B^*_i}w(b_t)/w(\B^*)$, and $\theta_i=\sum_{t\in\B^*_i}w(c_t)/w(\B^*)$. Recall that $\sum_{t\in\B^*_i}w(a_t)=\frac{u_i}{u_i+v_i+1}w(\B^*_i)=\frac{u_i}{u_i+v_i+1}\alpha_iw(\B^*)$.
Hence, we can get that $\rho_i=\frac{u_i}{u_i+v_i+1}\alpha_i$.
Analogously, we have $\sigma_i=\frac{v_i}{u_i+v_i+1}\alpha_i$ and $\theta_i=\frac{1}{u_i+v_i+1}\alpha_i$. Hence, we have
\begin{equation}\label{lp2}
\rho_i+\sigma_i+\theta_i=\alpha_i,\quad i\in\{1,2,3,4,5\}.
\end{equation}
By Lemmas \ref{t1}, \ref{t2}, \ref{t2+}, and \ref{t3+}, we can get that
\begin{equation}\label{lp3}
w(\T_1)/w(\B^*)\geq\alpha_1+2\rho_2+2\rho_3,
\end{equation}
\begin{equation}\label{lp4}
w(\T_2)/w(\B^*)\geq 2\theta_1+2\theta_2+2\theta_3+2\theta_4+2\theta_5,
\end{equation}
\begin{equation}\label{lp5}
w(\T_2)/w(\B^*)\geq\frac{2}{3}+\frac{97(1-3\tau)}{3645}\alpha_2+\frac{97(1-3\tau)}{1215}\alpha_4,
\end{equation}
\begin{equation}\label{lp6}
w(\T_3)/w(\B^*)\geq \frac{2}{3}(1-\varepsilon)+\frac{\tau}{36}\alpha_3+\frac{\tau}{36}\alpha_5.
\end{equation}
For each triangle $t$, we have $w(c_t)\geq w(b_t)\geq w(a_t)$ by the definition and $w(a_t)+w(b_t)\geq w(c_t)$ by the triangle inequality.
Hence, we also have
\begin{equation}\label{lp7}
\rho_i+\sigma_i\geq\theta_i\geq\sigma_i\geq\rho_i\geq 0,\quad i\in\{1,2,3,4,5\}.
\end{equation}
If $\tau$ is fixed, using (\ref{lp1})-(\ref{lp7}) the approximation ratio $\frac{\max\{w(\T_1),\ w(T_2),\ w(\T_3)\}}{w(\B^*)}$ can be obtained via solving a linear program (see Appendix~\ref{A.1}).

Setting $\tau=0.25$, we can get an approximation ratio of at least $(0.66835-\varepsilon)$. Hence, we have the following theorem.
\begin{theorem}
For MWMTP with any constant $\varepsilon>0$, there is a polynomial-time $(0.66835-\varepsilon)$-approximation algorithm.
\end{theorem}

\section{Conclusion}\label{conclusion}
In this paper, we consider approximation algorithms for the maximum weight metric triangle packing problem. This problem admits an almost-trivial $2/3$-approximation algorithm~\cite{hassin1997approximation1}.
The first nontrivial result, given by Chen et al.~\cite{DBLP:journals/jco/ChenCLWZ21}, is a randomized $(0.66768-\varepsilon)$-approximation algorithm.
Based on novel modifications, deep analysis, and conditional expectations, we propose a deterministic $(0.66835-\varepsilon)$-approximation algorithm. Whether it admits a simple algorithm with a better-than-$2/3$-approximation ratio is still unknown.

In our analysis of the first triangle packing, the internal edges of the partial-external triangles contained in the short cycle packing are all the least weighted edges. If we take care of this special case, we may also obtain some tiny improvements.

In the future, it would be interesting to study the well-related maximum weight metric 3-path packing problem, where we need to find a set of $n/3$ vertex-disjoint $3$-paths with the total weight maximized. This problem admits a similar almost-trivial $3/4$-approximation algorithm (see~\cite{li2023cyclepack}). However, it is still unknown to obtain a nontrivial approximation algorithm for this problem.

\section*{Acknowledgments}
The work is supported by the National Natural Science Foundation of China, under grant 62372095.

\section*{Declaration of competing interest} 
The authors declare that they have no known competing financial interests or personal relationships that could have appeared to influence the work reported in this paper.

\bibliographystyle{plain}
\bibliography{main}

\appendix

\section{The Linear Program}\label{A.1}
We consider that $\varepsilon$ is a very small constant. Then, the linear program can be built as follows.
\begin{alignat}{2}
\min\quad &y \nonumber \\
\mbox{s.t.}\quad 
&\alpha_1+\alpha_2+\alpha_3+\alpha_4+\alpha_5=1, \nonumber\\
&\rho_i+\sigma_i+\theta_i=\alpha_i,\quad i\in\{1,2,3,4,5\}, \nonumber\\
&y\geq\alpha_1+2\rho_2+2\rho_3, \nonumber\\
&y\geq 2\theta_1+2\theta_2+2\theta_3+2\theta_4+2\theta_5, \nonumber\\
&y\geq\frac{2}{3}+\frac{97(1-3\tau)}{3645}\alpha_2+\frac{97(1-3\tau)}{1215}\alpha_4, \nonumber\\
&y\geq\frac{2}{3}+\frac{\tau}{36}\alpha_3+\frac{\tau}{36}\alpha_5, \nonumber\\
&\rho_i+\sigma_i\geq\theta_i\geq\sigma_i\geq\rho_i\geq0,\quad i\in\{1,2,3,4,5\}. \nonumber
\end{alignat}
Setting $\tau=0.25$, the value of this LP is at least $0.668357$. Hence, we obtain an approximation ratio of $(0.66835-\varepsilon)$ for any constant $\varepsilon>0$.\footnote{In fact, we can directly obtain an approximation ratio of 0.66835 by setting $\varepsilon$ as a sufficiently small constant.}
\end{document}